\documentclass[11pt]{article}
\usepackage[pagebackref,colorlinks=true,pdfpagemode=none,linkcolor=blue,citecolor=blue,pdfstartview=FitH]{hyperref}

\usepackage[utf8]{inputenc}
\usepackage[russian,english]{babel}
\usepackage{amsmath}
\usepackage{amssymb}
\usepackage{amsthm}
\usepackage{fullpage}
\usepackage{mathpazo}
\usepackage{blkarray}
\usepackage{bm}
\usepackage[capitalize,nameinlink]{cleveref}
\renewcommand{\eqref}[1]{\hyperref[#1]{(\ref*{#1})}}

\theoremstyle{theorem}
\newtheorem{theorem}{Theorem}
\newtheorem{corollary}[theorem]{Corollary}
\newtheorem{lemma}[theorem]{Lemma}

\newtheorem{definition}[theorem]{Definition}
\newtheorem{claim}[theorem]{Claim}
\newtheorem{fact}[theorem]{Fact}
\theoremstyle{definition}
\newtheorem{remark}[theorem]{Remark}
\newtheorem{observation}[theorem]{Observation}

\newcommand{\prob}[2]{\mathop{\mathrm{Pr}}_{#1}\left[#2\right]}
\newcommand{\avg}[2]{\mathop{\textbf{E}}_{#1}[#2]}
\newcommand{\poly}{\mathop{\mathrm{poly}}}

\newcommand{\AC}{\mbox{\rm AC}}

\newcommand{\naturals}{\mathbb{N}}

\newcommand{\mc}[1]{\mathcal{#1}}

\newcommand{\OR}{\mathrm{OR}}
\newcommand{\AND}{\mathrm{AND}}

\newcommand{\Err}{\mathrm{Err}}

\title{On Polynomial Approximations to $\AC^0$\thanks{A preliminary version of this paper appeared in {\em Proc.\ $20$th International Workshop on Randomization
  and Computation (RANDOM)} 2016~\cite{HarshaS2016a}}}
\author{Prahladh Harsha\thanks{TIFR, Mumbai, India. \texttt{prahladh@tifr.res.in}. Research supported in part by UGC-ISF
grant 6-2/2014(IC).} 
\and
Srikanth Srinivasan\thanks{Department of Mathematics, IIT Bombay, Mumbai, India. \texttt{srikanth@math.iitb.ac.in}}}
\begin{document}
\maketitle

\begin{abstract}
Classical $\AC^0$ approximation results show that any $\AC^0$ circuit of size $s$ and depth $d$ has an $\varepsilon$-error probabilistic polynomial over the reals of degree $(\log (s/\varepsilon))^{O(d)}$. We improve this upper bound to $(\log s)^{O(d)}\cdot \log(1/\varepsilon)$, which is much better for small values of $\varepsilon$.

We then use this result to show that $(\log s)^{O(d)}\cdot \log(1/\varepsilon)$-wise independence fools $\AC^0$ circuits of size $s$ and depth $d$ up to error at most $\varepsilon$, improving on Tal's strengthening of Braverman's result that $(\log (s/\varepsilon))^{O(d)}$-wise independence suffices. To our knowledge, this is the first PRG construction for $\AC^0$ that achieves optimal dependence on the error $\varepsilon$.

We also prove lower bounds on the best polynomial approximations to $\AC^0$. We show that any polynomial approximating the $\OR$ function on $n$ bits to a small constant error must have degree at least $\widetilde{\Omega}(\sqrt{\log n})$. This result improves exponentially on a result of Meka, Nguyen, and Vu ({\em Theory Comput.} 2016).
\end{abstract}

\section{Motivation and Results}

In this paper, we study $\AC^0$ circuits, the family of circuits of constant depth and polynomial size (in the input length) with unbounded-fanin $\AND$ and $\OR$ gates. We use $\AC^0(s,d)$ to denote the family of $\AC^0$ circuits of size $s$ and depth $d$.

\paragraph{Polynomial approximations to $\AC^0$.} In his breakthrough work on proving lower bounds for the class $\AC^0[\oplus]$, Razborov~\cite{Razborov1987} studied how well small circuits can be approximated by low-degree polynomials. We recall (an equivalent version of) his notion of polynomial approximation over the reals.

An \emph{$\varepsilon$-error probabilistic polynomial} (over the reals) for a circuit $C(x_1,\ldots,x_n)$ is a random polynomial $\mathbf{P}(x_1,\ldots,x_n)\in \mathbb{R}[x_1,\ldots,x_n]$ such that for any $a\in \{0,1\}^n$, we have $\prob{\mathbf{P}}{C(a) \neq \mathbf{P}(a)}\leq \varepsilon$. Further, we say that $\mathbf{P}$ has degree $D$ and $\Vert \mathbf{P}\Vert_\infty \leq L$ if $\mathbf{P}$ is supported on polynomials $P$ of degree at most $D$ and $L_\infty$ norm at most $L$ (i.e. polynomials $P$ such that $\max_{a\in \{0,1\}^n}|P(a)| \leq L$). If there is such a $\mathbf{P}$ for $C$, we say that $C$ has an $\varepsilon$-error probabilistic degree at most $D$ and $L_\infty$ norm at most $L$.

It is well-known~\cite{TodaO1992,Tarui1993,BeigelRS1991} that any circuit $C\in \AC^0(s,d)$ has an $\varepsilon$-error probabilistic polynomial $\mathbf{P}$ of degree $(\log (s/\varepsilon))^{O(d)}$ and satisfying $\Vert \mathbf{P}\Vert_\infty < \exp\left((\log s/\varepsilon)^{O(d)}\right)$. This can be used to prove, for example~\cite{Smolensky1987}, (a slightly weaker version of) H\r{a}stad's theorem~\cite{Hastad1989} that says that Parity does not have subexponential-sized $\AC^0$ circuits. It also plays an important role in Braverman's theorem~\cite{Braverman2010} that shows that polylog-wise independence fools $\AC^0$ circuits.

\paragraph{Upper bounds for probabilistic polynomials.} We show a general result regarding error reduction of probabilistic polynomials over the reals. 

\begin{theorem}
\label{thm:gen-err}
Suppose $f:\{0,1\}^n\rightarrow \{0,1\}$ has a $(\frac{1}{2}-\delta)$-error probabilistic polynomial $\mathbf{P}$ of degree $D$ and $L_\infty$ norm at most $L\geq 2$. Then, for any $\varepsilon > 0$, $f$ has an $\varepsilon$-error probabilistic polynomial of degree at most $O\left(\frac{D}{\delta^2}\log(1/\varepsilon)\right)$ and $L_\infty$ norm at most $L^{O\left(\frac{1}{\delta^2}\log \frac{1}{\varepsilon}\right)}$.
\end{theorem}

Applying the above result to $(1/10)$-error probabilistic polynomials for $\AC^0$ gives us small-error probabilistic polynomials for $\AC^0$ with better parameters.

\begin{theorem}
\label{thm:prob-polys}
Let $C$ be any $\AC^0$ circuit of size $s$ and depth $d$. Let $\varepsilon > 0$ be any parameter. The circuit $C$ has an $\varepsilon$-error probabilistic polynomial $\mathbf{P}$ of degree at most $(\log s)^{O(d)}\cdot \log(1/\varepsilon)$ and $L_\infty$ norm at most $\exp\left((\log s)^{O(d)}\log (1/\varepsilon)\right)$.
\end{theorem}

Similar results on probabilistic polynomials were obtained over $\mathbb{F}_2$ (for the larger class of $\AC^0[\oplus]$ circuits) by Kopparty and Srinivasan~\cite{KoppartyS2012} and extended to all fixed non-zero characteristics by Oliveira and Santhanam~\cite{OliveiraS2015}. They have also found applications in the works of Williams~\cite{Williams2014}, for the purposes of obtaining better algorithms for integer programming, and Oliveira and Santhanam~\cite{OliveiraS2015}, for proving lower bounds on compression by bounded-depth circuits. However, as far as we know, no corresponding results were observed over the reals until now.

The above theorem was motivated by an application to constructing pseudorandom generators (PRGs) for $\AC^0$. As mentioned above, it was shown by Braverman~\cite{Braverman2010} that $\AC^0$ is fooled by polylog-wise independence. The proof of Braverman's theorem proceeds by constructing certain approximating polynomials for $\AC^0$, which in turn depends on two previous polynomial approximation results for this circuit class. The first of these is the $L_2$-approximation result of Linial, Mansour and Nisan~\cite{LinialMN1993} which is based on the classical H\r{a}stad Switching Lemma~\cite{Hastad1989}, and the second is the above mentioned result of Tarui~\cite{Tarui1993} and Beigel et al.~\cite{BeigelRS1991}. Using these constructions, Braverman showed that $\AC^0(s,d)$ is $\varepsilon$-fooled by $(\log (s/\varepsilon))^{O(d^2)}$-wise independence. 

An example due to Mansour appearing in the work of Luby and Veli\v{c}kovi\'{c}~\cite{LubyV1996} demonstrated that $(\log s)^{d-1}\log(1/\varepsilon)$-wise independence is \emph{necessary} to $\varepsilon$-fool $\AC^0(s,d)$. This leads naturally to the question of showing tight bounds for the amount of independence required to fool $\AC^0(s,d)$.

Using an improved switching lemma due to H\r{a}stad~\cite{Hastad2014} (see also the work of Impagliazzo, Matthews, and Paturi~\cite{ImpagliazzoMP2012}), Tal~\cite{Tal2017} gave an improved version of the $L_2$-approximation result of Linial et al.~\cite{LinialMN1993}, and used this to improve the parameters of Braverman's theorem. Specifically, he showed that $(\log (s/\varepsilon))^{O(d)}$-wise independence fools $\AC^0$. 

Tal asked if the dependence on $\varepsilon$ in this result could be made to match the limit given  by Mansour's example. Formally, he asked if $(\log s)^{O(d)}\cdot\log(1/\varepsilon)$-wise independence fools $\AC^0(s,d)$. In this work, we are able to answer this question in the affirmative (\cref{cor:PRGs} below). Up to the constant implicit in the $O(d)$, our result is optimal for all $\varepsilon > 0$.

\paragraph{Comparison to other PRGs for $\AC^0$.} Using standard constructions of $k$-wise independent probability distributions, the above result gives explicit PRGs with seedlength $(\log s)^{O(d)}\cdot \log(1/\varepsilon)$ for fooling circuits from $\AC^0(s,d)$. It is easy to see that this seedlength cannot be improved beyond $\Omega(\log(1/\varepsilon))$ and hence that our result is optimal in terms of the error parameter $\varepsilon$. 

It is also instructive to see how well this compares to general (i.e. not based on limited independence) PRG constructions for $\AC^0$. Using the standard Hardness-to-Randomness paradigm of Nisan and Wigderson~\cite{NisanW1994} and the best known average case lower bounds for $\AC^0$~\cite{ImpagliazzoMP2012,Hastad2014}, it is easy to obtain PRGs of seedlength $(\log s)^{O(d)}\cdot(\log(1/\varepsilon))^2$ for $\AC^0(s,d)$. Furthermore, the Nisan-Wigderson paradigm cannot yield PRGs of seedlength less than $(\log(1/\varepsilon))^2$ given our current state of knowledge regarding circuit lower bounds (see \cref{app:NW} for details). Another recent PRG construction for $\AC^0(s,d)$ due to Trevisan and Xue~\cite{TrevisanX2013} has seedlength $(\log (s/\varepsilon))^{d+O(1)}$. 

The reader will note that both constructions are suboptimal in terms of the dependence on $\varepsilon$ (though both are better than ours in terms of dependence on $s$ and $d$). Interestingly, as far as we know, our construction is the first that achieves an optimal dependence on $\varepsilon$.

\paragraph{Lower bounds for probabilistic polynomials.} We can also ask if our result can be strengthened to yield a seedlength of $(\log s)^{d+O(1)}\cdot \log (1/\varepsilon)$, which would generalize both our current construction and that of Trevisan and Xue~\cite{TrevisanX2013}, and almost  match Mansour's lower bound as well. Such a strengthening could conceivably be obtained by improving the polynomial approximation results for $\AC^0$~\cite{Tarui1993,BeigelRS1991}. Razborov~\cite{Razborov1987} observed that to obtain good approximations for $\AC^0(s,d)$, it suffices to approximate the $\OR$ function on $s$ bits efficiently.  Therefore, we study the probabilistic degree of the $\OR$ function.

Beigel, Reingold and Spielman~\cite{BeigelRS1991} and Tarui~\cite{Tarui1993} showed that the OR function on $n$ bits can be $\varepsilon$-approximated by a polynomial of degree $O((\log n)\cdot \log(1/\varepsilon))$.
While it is easy to show that the dependence on $\varepsilon$ in this result is tight (in fact for \emph{any field}), for a long time, it was not known if \emph{any} dependence on $n$ is necessary over the reals\footnote{In fact, for finite fields of constant size, Razborov~\cite{Razborov1987} showed that the $\varepsilon$-error probabilistic degree of $\OR$ is $O(\log(1/\varepsilon))$, independent of the number of input bits.}. Recently, Meka, Nguyen and Vu~\cite{MekaNV2016} showed that any \emph{constant error} probabilistic polynomial for the OR function over the reals must have degree $\widetilde{\Omega}(\log \log n)$ and hence the dependence on the parameter $n$ is unavoidable. We further improve the bound of Meka et al. exponentially to $\widetilde{\Omega}(\sqrt{\log n})$, which is only a quadratic factor away from the upper bound. 

\subsection{Proof ideas}

Here, we describe the ideas behind the proofs of the main results. 

The proof of \cref{thm:gen-err} is extremely simple. A natural strategy to reduce the error of a (constant-error, say) probabilistic polynomial $\mathbf{P}$ is to sample it independently $\ell = O(\log (1/\varepsilon))$ times to obtain polynomials $\mathbf{P}_1,\ldots,\mathbf{P}_\ell$ and then take the Majority vote among the $\mathbf{P}_i$s, which can be simulated by composing with a multilinear polynomial of degree $\ell$. Indeed, this is exactly what Kopparty and Srinivasan~\cite{KoppartyS2012} do in an earlier work to obtain $\varepsilon$-error probabilistic polynomials over $\mathbb{F}_2$. 

Over the reals, it is not completely clear that this strategy works, since the polynomials $\mathbf{P}_i$ need not output a Boolean value when they err and hence it is not clear what taking a ``Majority vote'' means. Nevertheless, we observe that composing with the multilinear Majority polynomial continues to work since this polynomial has the nice property that setting more than half of its input bits to a constant $b\in \{0,1\}$ causes the polynomial to collapse to the constant polynomial $b$, which is oblivious to the values of the unset inputs (that could even be non-Boolean  and possibly arbitrarily large real numbers). 

As mentioned already above, \cref{thm:gen-err}, along with standard constructions of probablistic polynomials for $\AC^0(s,d)$ in the constant-error regime, directly proves \cref{thm:prob-polys}. We can more or less plug this result into Tal's proof~\cite{Tal2017} of Braverman's theorem to obtain better parameters for the amount of independence required to fool $\AC^0$. The only additional idea required is to ensure that the inputs where the probabilistic polynomial computes the correct value are certified by a small $\AC^0$ circuit. While a small $\AC^0$ circuit \emph{cannot} compute the Majority vote above, it turns out that a weaker ``Approximate Majority'' (see Definition~\ref{def:approx-maj} below) is sufficient for this purpose, and this can be done in $\AC^0$, as shown by Ajtai and Ben-Or\cite{AjtaiB1984}.

We now describe the proof of the degree lower bound for $\varepsilon$-error probabilistic polynomials computing the $\OR$ function on $n$ variables to a small constant-error (say $1/10$). It is known that this can be done over fields of \emph{constant} characteristic with constant degree~\cite{Razborov1987} and over the reals with degree $O(\log n)$~\cite{Tarui1993,BeigelRS1991}. Hence any technique for proving lower bounds growing with $n$ will have to use a technique specific to large characteristic.

The work of Razborov and Viola~\cite{RazborovV2013} introduced such a technique to the theoretical computer science literature to show that no low degree polynomial over the reals can compute the Parity function on more than half its inputs. The main technique was an anti-concentration lemma generalizing  classical theorems of Littlewood-Offord and Erd\H{o}s~\cite{LittlewoodO1938,Erdos1945} that state that any linear function of at least $r$ Boolean variables takes any fixed value on a uniformly random input with probability at most $O(1/\sqrt{r})$. In particular, it cannot approximate a Boolean function well unless $r$ is very small. Razborov and Viola, building on the work of Costello, Tao, and Vu~\cite{CostelloTV2006}, proved a generalization of this statement to low-degree multivariate polynomials that contain at least $r$ disjoint monomials of maximum degree. 

More recently, Meka, Nguyen, and Vu~\cite{MekaNV2016} proved an improved (and near-optimal) version of the anti-concentration lemma of Razborov and Viola and used this to show better lower bounds for the Parity function. Additionally, they were also able to show that any constant-error probabilistic polynomial for the $\OR$ function must have degree $\widetilde{\Omega}(\log \log n)$. We use their anti-concentration lemma with a more efficient restriction argument to prove a lower bound of $\widetilde{\Omega}(\sqrt{\log n})$. We describe the outline of this restriction argument next.

To prove a lower bound of $D$ on the probabilistic degree of some function it suffices (and is also necessary, by standard duality arguments) to obtain a distribution under which the function is hard to approximate by any polynomial of degree less than $D$. While some functions have `obvious' hard distributions (such as the Parity function, which is random self-reducible w.r.t. the uniform distribution), the $\OR$ function is not one such, since it takes value $0$ only on one input. Some obvious candidates (such as the uniform distribution or a convex combination of the uniform distribution along with the distribution that puts all its mass on the all $0$s input) can actually be shown to be easy for the $\OR$ function. The hard distribution we use is motivated by the polynomial constructions of~\cite{BeigelRS1991,Tarui1993} and is as follows: with probability $1/2$ choose the all $0$s input and with probability $1/2$ choose a uniformly random $i\in [\log n]$ and then choose a random input of weight\footnote{We will actually use the product distribution where each bit is set to $1$ with probability $\frac{1}{2^i}$, which puts most of its mass on inputs of weight \emph{close} to $n/2^i$, but we blur this distinction here.} $n/2^i$. The hard distribution chosen by Meka et al. is similar, but sparser than the distribution we use (it is only concentrated on $\log \log n$ levels of the hypercube whereas our distribution is concentrated on $\log n$ levels). 

We now argue that any polynomial $q$ approximating the $\OR$ function w.r.t. this distribution must be of large degree as follows. First of all, since there is a considerable amount of mass on the all $0$s input, we can assume that $q$ takes value $0$ on this input. Now, we consider the distribution that is uniformly distributed on inputs of Hamming weight $n/2$. We know that the $\OR$ function is always $1$ on these inputs, which means that 
$q$ is \emph{not} anti-concentrated on inputs from this distribution (since it must take the value $1$ most of the time). Hence, by the anti-concentration lemma due to Meka et al., any maximal disjoint set of maximum degree monomials in $q$ cannot have too many monomials, say more than $r$. In particular, setting all the variables $V$ in such a set of monomials --- there are at most $rD$ variables in $V$ --- to $0$ reduces the degree of the polynomial by $1$. The important observation is that this naturally happens with high probability when we use the distribution that is uniformly distributed on inputs of weight $\approx n/rD$, since each variable is set to $1$ only with probability $\approx 1/rD$. Further, we can simulate the uniform distribution on (say) inputs of weight $n/rD$ by first sampling a set $S$ of size $2n/rD$ and setting the bits outside $S$ to $0$ --- this sets all the variables in $V$ with good probability and thus reduces the degree of $q$ --- and then choosing a random set of $|S|/2$ inputs to set to $1$. We are now exactly in the situation we were at the beginning of this paragraph, except for the fact that the degree of $q$ is smaller.

Continuing in this way, we eventually obtain a constant polynomial $q$ that computes the $\OR$ function on some non-zero inputs from the hypercube, which means that it must be the constant polynomial $1$. However, this contradicts the fact that $q$ takes value $0$ on the all $0$s input and this proves the theorem.

\section{Improved probabilistic polynomials and PRGs for $\AC^0$}

\subsection{The construction of probabilistic polynomials}

\paragraph{Notation.} Let $P\in \mathbb{R}[x_1,\ldots,x_\ell]$. Given a set $S\subseteq [\ell]$ and a partial assignment $\sigma:S\rightarrow \{0,1\}$, we define $P|_{\sigma}$ to be the polynomial obtained by setting all the bits in $S$ according to $\sigma$. In the case that $\sigma$ sets all the variables in $S$ to a constant $b\in \{0,1\}$, we use $P|_{S\mapsto b}$ instead of $P|_\sigma$. For a function $f:\{0,1\}^\ell\rightarrow \{0,1\}$, we define $f|_\sigma$ and $f|_{S\mapsto b}$ similarly.

We define the \emph{weight} of $P$, denoted $w(P)$, to be the sum of the absolute values of all the coefficients of $P$.

\begin{definition}
\label{def:pseudo-maj}
Let $P\in\mathbb{R}[x_1,\ldots,x_\ell]$ and say $r$ is a parameter from $[\ell]$. We say that $P$ is an $\ell$-\emph{pseudo-majority} if for $r$ being the least integer greater than $\ell/2$ and any $S\in \binom{[\ell]}{r}$ and $b\in \{0,1\}$, the polynomial $P|_{S\mapsto b}$ is the constant polynomial $b$.
\end{definition}	

We show below that the multilinear polynomial representing the Majority function is an $\ell$-pseudo-majority of weight $2^{O(\ell)}$.


Before we prove that this construction works, we need a few standard facts about polynomials.

\begin{fact}
\label{fac:polys}
Any Boolean function $f:\{0,1\}^\ell\rightarrow \{0,1\}$ can be represented uniquely by a multilinear polynomial $P[x_1,\ldots,x_\ell]$ in the sense that for all $a\in \{0,1\}^n$, we have $P(a) = f(a)$. Furthermore, $w(P)= 2^{O(\ell)}$.
\end{fact}

The uniqueness in the fact above yields the following  observation. 

\begin{lemma}
\label{lem:cert}
Let $f:\{0,1\}^\ell \rightarrow\{0,1\}$ and $P$ be the corresponding unique multilinear polynomial guaranteed by \cref{fac:polys}. If $\sigma:S\rightarrow \{0,1\}$ is a partial assignment such that $f|_\sigma$ is the constant function $b\in \{0,1\}$, then $P|_\sigma$ is formally the constant polynomial $b$.
\end{lemma}

\begin{proof}
Follows from the fact that $P|_\sigma$ is a multilinear polynomial representing the constant function $b$ on the variables not in $S$ and the uniqueness part of \cref{fac:polys}.
\end{proof}

\begin{remark}
\label{rem:anyreal}
Note that the hypothesis of the lemma above is that $f|_\sigma(a) = b$ for all Boolean assignments $a$ to the remaining variables. However, the conclusion yields a stronger conclusion for the polynomial $P$: namely, we show that $P|_\sigma$ takes value $b$ on \emph{any} assignment $a\in \mathbb{R}^{\ell-|S|}$ to the remaining variables, and not just Boolean assignments. It is this fact that we will use in applications below.
\end{remark}

For $\ell\in\mathbb{N}$, define the Boolean function $M_\ell$ to be the Majority function: i.e., $M_\ell(x) = 1$ iff the Hamming weight of $x$ is strictly greater than $\ell/2$. Note that for any $S\subseteq [\ell]$ of size greater than $\ell/2$ and any $b\in \{0,1\}$, $M_\ell|_{S\mapsto b}$ is the constant function $b$.

Let $P_\ell$ be the multilinear polynomial representing $M_\ell$ guaranteed by \cref{fac:polys}. Applying \cref{lem:cert} to the pair $M_\ell$ and $P_\ell$, we obtain the following corollary.

\begin{corollary}
\label{cor:pseudo-maj}
For any $\ell\in\naturals$, there exist $\ell$-pseudo-majorities of degree $\ell$ and weight $2^{O(\ell)}$.
\end{corollary}

We now prove \cref{thm:gen-err}. We will follow   the proof of~\cite[Lemma 10]{KoppartyS2012}, but some additional justification will be required since we are working over the reals and not over $\mathbb{F}_2$ as in \cite{KoppartyS2012}.

\begin{proof}[Proof of \cref{thm:gen-err}] 
We set $\ell = \frac{A}{\delta^2}\log(\frac{1}{\varepsilon})$ for a constant $A>0$ to be fixed later. Let $\mathbf{P}_1,\ldots,\mathbf{P}_\ell$ be $\ell$ mutually independent copies of the probabilistic polynomial $\mathbf{P}$. Fix an $\ell$-pseudo-majority $Q$ as guaranteed by \cref{cor:pseudo-maj}. The final probabilistic polynomial is $\mathbf{R} = Q(\mathbf{P}_1,\ldots,\mathbf{P}_\ell)$. 

The degree of $\mathbf{R}$ is at most $\deg(Q)\cdot \deg(\mathbf{P}) \leq O(\frac{D}{\delta^2}\log(\frac{1}{\varepsilon}))$. Moreover, it can be seen that the $\Vert \mathbf{R}\Vert_\infty\leq w(Q)\cdot L^{\deg(Q)} \leq (2L)^{O(\ell)}\leq L^{O(\ell)}$ since $L\geq 2$. 

Finally, we see that for any $a\in \{0,1\}^n$, $\mathbf{R}(a) = f(a)$ unless it holds that for at least $\lfloor \ell/2\rfloor$ many $i\in [\ell]$, we have $\mathbf{P}_i(a)\neq f(a)$. By a Chernoff bound, the probability of this is at most $\varepsilon$ as long as $A$ is chosen to be a suitably large constant. Hence, $\mathbf{R}$ is indeed an $\varepsilon$-error probabilistic polynomial for $f$. 
\end{proof}

\cref{thm:prob-polys} immediately follows from the above and standard probabilistic polynomials for $\AC^0$ from~\cite{TodaO1992,Tarui1993,BeigelRS1991}. However, for our applications to PRGs for $\AC^0$, we need a slightly stronger statement, which we prove below.

\newcommand{\bmc}[1]{\bm{\mc{#1}}}
\begin{definition}[Probabilistic polynomial with witness]
An $\varepsilon$-error probabilistic polynomial for circuit $C(x_1,\ldots,x_n)$ with witness ($\varepsilon$-error PPW for short) is a pair $(\mathbf{P},\bm{\mc{E}})$ of random variables such that $\mathbf{P}$ is a randomized polynomial  and $\bm{\mc{E}}$ is a randomized circuit (both on $n$ Boolean variables) such that for any input $a\in \{0,1\}^n$, we have
\begin{itemize}
\item  $\prob{\bmc{E}}{\bmc{E}(a) = 1} \leq \varepsilon$,
\item For any fixing $(P,\mc{E})$ of $(\mathbf{P},\bmc{E})$, we have $\mc{E}(a) = 0 \Rightarrow P(a) = C(a)$.
\end{itemize}
In particular, this implies that $\mathbf{P}$ is an $\varepsilon$-error probabilistic polynomial for $C$. 

We say that $\bmc{E}$ belongs to a circuit class $\mc{C}$ if it is supported on circuits from class $\mc{C}$.
\end{definition}

The above notion was introduced in Braverman~\cite{Braverman2010} who proved the following lemma, building on earlier works of~\cite{TodaO1992,Tarui1993, BeigelRS1991}.
\begin{lemma}[{\cite[Lemma 8, Proposition 9]{Braverman2010}}]
\label{lem:Tarui-Brav}
Fix parameters $s,d\in\naturals$ and $\varepsilon > 0$. Any $\AC^0$ circuit $C$ of size $s$ and depth $d$ has an $\varepsilon$-error PPW $(\mathbf{P},\bmc{E})$ where
\begin{itemize}
\item $\deg(\mathbf{P}) \leq (\log(s/\varepsilon))^{O(d)}$ and $\Vert \mathbf{P}\Vert_\infty \leq \exp\left((\log (s/\varepsilon)\right)^{O(d)})$,
\item $\bmc{E}\in \AC^0(\poly(s\log(1/\varepsilon)),d+3)$.
\end{itemize}
\end{lemma}

We show the following variant of the above lemma, which is an improvement in terms of degree and the $L_\infty$ norm of the probabilistic polynomial for small $\varepsilon$.
\begin{lemma}
\label{lem:PPW}
Fix parameters $s,d\in\naturals$ and $\varepsilon > 0$. Any $\AC^0$ circuit $C$ of size $s$ and depth $d$ has an $\varepsilon$-error PPW $(\mathbf{P},\bmc{E})$ where
\begin{itemize}
\item $\deg(\mathbf{P}) \leq (\log s)^{O(d)}\cdot \log(1/\varepsilon)$ and $\Vert \mathbf{P}\Vert_\infty \leq \exp\left((\log s)^{O(d)}\log(1/\varepsilon)\right)$,
\item $\bmc{E}\in \AC^0(\poly(s\log(1/\varepsilon)),d+O(1))$.
\end{itemize}
\end{lemma}

Before we begin the proof, we state one more result from the literature. 

\begin{definition}
\label{def:approx-maj}
Given an integer parameter $\ell$ and real parameters $\alpha, \beta \in [0,1]$ with $\alpha < \beta$, we will call a function $f:\{0,1\}^\ell\rightarrow \{0,1\}$ an \emph{$(\ell,\alpha,\beta)$-approximate majority} if $f(x) = 0$ for any input of Hamming weight at most $\alpha\ell$ and $f(x) = 1$ for any input of Hamming weight at least $\beta\ell$. 
\end{definition}

The following is a result of Ajtai and Ben-Or~\cite{AjtaiB1984}.

\begin{lemma}[Ajtai and Ben-Or~\cite{AjtaiB1984}]
\label{lem:appxmaj}
Fix any constants $\alpha < \beta$. Then, for all $\ell\in\naturals$, there is an $(\ell,\alpha,\beta)$-approximate majority which has an $\AC^0$ circuit of size $\poly(\ell)$ and depth $3$.
\end{lemma}

We now prove \cref{lem:PPW}. The proof is similar to that of \cref{thm:gen-err} above, but we also need to obtain a witness circuit for our probabilistic polynomial.

\begin{proof}[Proof of \cref{lem:PPW}]
Let $\ell = A\log(1/\varepsilon)$ for a large constant $A$ to be chosen later. W.l.o.g. assume that $\ell$ is even. Let $Q(x_1,\ldots,x_\ell)$ be the $\ell$-pseudo-majority guaranteed by \cref{cor:pseudo-maj}. By \cref{lem:appxmaj}, there is an $\AC^0$ circuit $C_1$ of size $\poly(\ell)$ and depth $3$ that computes an $(\ell,1/4,2/5)$-approximate majority.

Let $(\mathbf{P}_1,\bmc{E}_1),\ldots,(\mathbf{P}_\ell,\bmc{E}_\ell)$ be \emph{independent} copies of the $(1/8)$-error PPW guaranteed by \cref{lem:Tarui-Brav}. The final PPW is $(\mathbf{P},\bmc{E})$ where $\mathbf{P} = Q(\mathbf{P}_1,\ldots,\mathbf{P}_\ell)$ and $\bmc{E} = C_1(\bmc{E}_1,\ldots,\bmc{E}_\ell)$. We show that this PPW has the required properties. 

First of all, we know that on any input $a$ to the circuit $C$ and for any $i\in [\ell]$, the probability that $\bmc{E}_i(a) = 1$ is at most $1/8$. Thus, the expected number of $\bmc{E}_i$ that output $1$ is at most $\ell/8$. However, for $\bmc{E}(a)$ to be $1$, at least $\ell/4$ many $\bmc{E}_i(a)$ should be $1$. By a Chernoff bound, the probability of this event is at most $\exp(-\Omega(\ell)) < \varepsilon$ for a large enough constant $A$.

Now, we need to argue that if $\bmc{E}(a) = 0$, then $\mathbf{P}(a) = Q(\mathbf{P}_1(a),\ldots,\mathbf{P}_\ell(a)) = C(a)$. Say $C(a) = b\in \{0,1\}$. If $\bmc{E}(a) =0$, then we know that the number of $\bmc{E}_i(a)$ that are $0$ is at least $3\ell/5$; let $I$ denote the set of these $i$. By the definition of PPWs, we know that for each $i\in I$, we have $\mathbf{P}_i(a) = b$ and hence at least $3\ell/5 > \ell/2$ many inputs of $Q$ are set to $b$. Since $Q$ is an $\ell$-pseudo-majority, we must have $Q(\mathbf{P}_1(a),\ldots,\mathbf{P}_\ell(a)) = b$. This concludes the proof that $(\mathbf{P},\bmc{E})$ is indeed an $\varepsilon$-error PPW for $C$.

Note that $\deg(\mathbf{P})\leq \deg(Q)\cdot \max_i \deg(\mathbf{P}_i) \leq (\log s)^{O(d)}\log(1/\varepsilon)$. Also, it can be seen that
\[
\Vert \mathbf{P}\Vert_\infty \leq w(Q)\cdot(\max_{i\in [\ell]} \Vert \mathbf{P}_i \Vert_\infty)^{\deg(Q)} \leq \exp\left((\log s)^{O(d)}\log(1/\varepsilon)\right).
\]
Thus, $\mathbf{P}$ has the required properties. The size and depth properties of $\bmc{E}$ follow trivially from its definition. This concludes the proof of the lemma.
\end{proof}

\subsection{Application to PRGs for $\AC^0$}

The connection between probabilistic polynomials and PRGs for $\AC^0$ is encapsulated in the following theorem (which is an easy observation from the works of Braverman and Tal):
\begin{theorem}[Braverman~\cite{Braverman2010},Tal~\cite{Tal2017}]
\label{thm:bravtal}
Let $s,d\in\naturals$ and $\varepsilon > 0$. Suppose that any $\AC^0$ circuit of size $s$ and depth $d$ has an $(\varepsilon/2)$-error PPW $(\mathbf{P},\bmc{E})$ such that
\begin{itemize}
\item $\deg(\mathbf{P}) = D$, $\Vert \mathbf{P}\Vert_\infty \leq L$,
\item $\bmc{E}\in \AC^0(s_1,d_1)$, 
\end{itemize}
Then, $\AC^0$ circuits of size $s$ and depth $d$ can be $\varepsilon$-fooled by $k(s,d,\varepsilon)$-wise independence, where 
\[
k(s,d,\varepsilon) = O(D) + (\log s_1)^{O(d_1)}\cdot (\log (1/\varepsilon) + \log L)
\]
\end{theorem}

Note that the theorem above is trivial when $\log(1/\varepsilon) > s$ since any $\AC^0$ circuit of size $s$ is trivially fooled by an $s$-wise independent distribution. Hence, the theorem is non-trivial only when $\log(1/\varepsilon) \leq s$. In this case, using \cref{lem:PPW} and the theorem above, we immediately get

\begin{corollary}\label{cor:PRGs}
Fix parameters $s,d\in\naturals$ and $\varepsilon > 0$. Any circuit $C\in \AC^0(s,d)$ can be $\varepsilon$-fooled by any distribution that is $(\log s)^{O(d)}\log(1/\varepsilon)$-wise independent.
\end{corollary}

\section{The probabilistic degree of OR}

\paragraph{Notation.} For $i\geq 1$ and a set of Boolean variables $X$, let $\mu_i^X$ be the product distribution on $\{0,1\}^X$ defined so that for each $x\in X$, the probability that $x = 1$ is $2^{-i}$. We also use $\mc{U}_X$ to denote $\mu_1^X$, the uniform distribution over $\{0,1\}^X$. The OR function on the variables in $X$ is denoted $\OR_X$.

We want to show:

\begin{theorem}
\label{thm:OR-lbd}
Assume $|X_0| = n$. The $1/8$-error probabilistic degree of $\OR_{X_0}$ is $\Omega(\frac{\sqrt{\log n}}{(\log \log n)^{3/2}})$.
\end{theorem}

\begin{remark}
\label{rem:1/8}
Though the theorem is stated for error $1/8$, it is not hard to see that it holds (with constant factor losses) as long as the error is bounded by $1/2 - \Omega(1)$. One way to see this is to appeal to \cref{thm:gen-err}. Another way is to do a simpler error reduction specific to the $\OR$ function as we do in the proof of \cref{thm:OR-lbd}.
\end{remark}

In order to prove \cref{thm:OR-lbd}, we use an anti-concentration lemma due to Meka, Nguyen and Vu~\cite{MekaNV2016}\footnote{The result of Meka et al. is actually stated for polynomials over the \emph{Fourier basis} of Parity functions (see, e.g., the book of O'Donnell~\cite{ODonnell}). However, it is an easy observation that a polynomial of degree $d$ has $r$ disjoint terms of degree $d$ in the standard monomial basis if and only if it has $r$ disjoint terms of degree $d$ in the Fourier basis. Hence, the result holds in the standard basis as well.} coupled with a random restriction argument inspired by the work of Razborov and Viola~\cite{RazborovV2013}.

\begin{lemma}[Meka, Nguyen, and Vu~{\cite[Theorem~1.6]{MekaNV2016}}]
\label{lem:NV}
There exists an absolute constant $B > 0$ so that the following holds. Let $p(x)\in \mathbb{R}[X]$ be a degree $d$ multilinear polynomial with at least $r$ disjoint degree $d$ terms. Then $\prob{x\sim \mc{U}_X}{p(x) = 0} \leq B d^{4/3} r^{-\frac{1}{4d+1}}\sqrt{\log r}$.
\end{lemma}

Note that the above lemma is a non-trivial statement only when $r = d^{\Omega(d)}.$ 

Given a polynomial $q\in \mathbb{R}[X]$, we denote by $\Err_i^X(q)$ the error of polynomial $q$ w.r.t. distribution $\mu_i^X$. Formally,
\[
\Err_i^X(q) = \prob{x \sim \mu_i^X}{q(x) \neq \OR_X(x)}
\]

For a set of variables $X$, $\ell\in \mathbb{N}$ and $\delta\in \mathbb{R}^{\geq 0}$, call a polynomial $q\in \mathbb{R}[X]$ $(X,\ell,\delta)$-good if 
\[
\avg{i\in [\ell]}{\Err_i^X(q)} \leq \delta.
\]

\begin{definition}
\label{def:restriction}
A \emph{random zero-fixing restriction} on the variable set $X$ with $*$-probability $p\in [0,1]$ will be a function $\rho:X\rightarrow \{*,0\}$ with each variable set independently to $*$ with probability $p$ and to $0$ otherwise. We use $X_\rho$ to denote $\rho^{-1}(*)$. The restriction of a polynomial $q$ under $\rho$ is denoted $q|_\rho$.
\end{definition}

\begin{observation}
\label{obs:restr-err}
Let $q\in \mathbb{R}[X]$ and $\rho$ be a zero-fixing random restriction on the variable set $X$ with $*$-probability $p = \frac{1}{2^b}$ where $b\in \mathbb{N}$. For any $i\geq 1$, 
\[
\avg{\rho}{\Err_i^{X_\rho}(q|_\rho)} = \Err^X_{i + b}(q)
\]
\end{observation}

(I.e., setting bits independently to $1$ with probability $\frac{1}{2^{i+b}}$ is the same as first applying a random zero-fixing restriction with $*$-probability $\frac{1}{2^{b}}$ and then setting each surviving variable to $1$ with probability $\frac{1}{2^{i}}$.)

\subsection{Proof of \cref{thm:OR-lbd}}

We argue by contradiction. Let $\mathbf{P}$ be a $1/8$-error probabilistic polynomial for $\OR_{X_0}$ of degree $D < \sqrt{\log n}/A(\log \log n)^{3/2}$ for some absolute constant $A>0$ that we will fix in \cref{clm:calc}. In particular, we have
\[
\prob{\mathbf{P}}{\mathbf{P}(0,0,\ldots,0) \neq 0} \leq \frac{1}{8}
\]

We discard all polynomials $q$ such that $q(0,0,\ldots,0) \neq 0$ from the distribution underlying $\mathbf{P}$ (i.e. we simply condition the distribution on not sampling such a polynomial). The resulting probabilistic polynomial $\mathbf{P}'$ is supported only on polynomials $q\in \mathbb{R}[X_0]$ such that $q(0,0,\ldots,0) = 0$ and further, it can be seen that $\mathbf{P}'$ is a $(1/4)$-error probabilistic polynomial for $\OR_{X_0}$ of degree $D$. 

Let $\mathbf{P}'_1,\ldots,\mathbf{P}'_s$ be $s = \log \log n$ independent instances of $\mathbf{P}'$ and let $\mathbf{Q} = 1-\prod_{i\in [s]}(1-\mathbf{P}'_i)$. Then, $\mathbf{Q}$ is an error $\frac{1}{4^s} = \frac{1}{\log^2 n}$ probabilistic polynomial for $\OR_n$ of degree at most $sD  < \sqrt{\log n}/A\sqrt{\log \log n}$. In particular, there is a polynomial $q_0\in \mathbb{R}[x_1,\ldots,x_n]$ of degree $d_0 < \sqrt{\log n}/A\sqrt{\log \log n}$ such that $q_0(0,0,\ldots,0) = 0$ and for $\varepsilon_0 = \frac{1}{\log^2 n}$ we have 
\[
\avg{i\in [(\log n)/2]}{\Err_i^{X_0}(q_0)} \leq \varepsilon_0
\]
Define $n_0 = |X_0| = n$ and $ \ell_0 = (\log n)/2$. By the above inequality, the polynomial $q_0$ is $(X_0, \ell_0,  \varepsilon_0)$-good. Also define parameters $r = (d_0\cdot \log^2 n)^{10d_0}$ and $p = \frac{1}{2^b}$ where $b\in \mathbb{N}$ is chosen so that $p \in [\frac{1}{2r^2},\frac{1}{r^2}]$. Note that 
\[r \leq (\log n)^{O(d_0)} \leq (\log n)^{O(\sqrt{\log n})} = n^{o(1)}\]
 and hence $p =\Theta(1/r^2) = 1/n^{o(1)}$.

We now define a sequence of polynomials $q_1,q_2,\ldots,q_t$ such that:
\begin{itemize}
\item Each $q_{i}\in \mathbb{R}[X_i]$ where $X_i\subseteq X_0$ and has degree $d_{i} \geq 0$. Also, $|X_i| = n_i$ where $n_i \in [pn_{i-1}/2, 3pn_{i-1}/2]$. Further $\deg(q_i) = d_i < d_{i-1}$. The polynomial $q_i = q_{i-1}|_{\rho_i}$ for some restriction $\rho_i:X_{i-1}\rightarrow\{*,0\}$.
\item Each polynomial $q_{i}$ is $(X_i, \ell_i, \varepsilon_i)$-good where $\ell_i = \ell_{i-1} - b$ and $\varepsilon_i = \varepsilon_{i-1}\cdot \exp\left(\frac{16b}{\log n}\right)$.
\item $d_t = \deg(q_t) = 0$. That is, $q_t$ is a constant polynomial.
\end{itemize}

Before we describe how to construct this sequence, let us see how it implies the desired contradiction. Note that since $d_i < d_{i-1}$ for each $i\geq 1$, the length $t$ of the sequence is bounded by $d_0 < \sqrt{\log n}/A\sqrt{\log \log n}$. 

We first make the following simple claim.

\begin{claim}
\label{clm:calc}
There is a large enough constant $A$ in the definition of $D$ above so that for each $i\in [t]$, $n_i\geq \sqrt{n}$, $\ell_i \geq \frac{\log n}{4}$, and $\varepsilon_i < \frac{1}{\log n}$.
\end{claim}
\begin{proof} 
It can be checked that the following inequalities hold for a large enough choice of the constant $A$. 

Firstly,
\[
n_i \geq n_t \geq n_0 \cdot (p/2)^t = n\cdot (d_0\log^2 n)^{-O(d_0^2)}\geq \sqrt{n}.
\]

Also, note that $\ell_i = \ell_0 - bi \geq \ell_0-bt = (\log n)/2 - O(d_0^2 \log \log n) \geq \frac{\log n}{4}$ and 
\[
\varepsilon_i  = \varepsilon_0\cdot \exp\left(\frac{16bi}{\log n}\right) \leq \varepsilon_0\cdot \exp\left(\frac{16bt}{\log n}\right) = \frac{1}{\log^2 n}\cdot \exp\left(\frac{O(d_0^2\log\log n)}{\log n}\right) < \frac{1}{\log n}.
\]

\end{proof}

In particular, since $q_t$ is $(X_t,\ell_t,\varepsilon_t)$-good, we must have 
\begin{equation}
\label{eq:qt-err}
\Err^{X_t}_1(q_t) \leq \ell_t \avg{i\in [\ell_t]}{\Err^{X_t}_i(q_t)} < \varepsilon_t \ell_t < \frac{1}{2}
\end{equation}
using the fact that $\ell_t \leq \ell_0 = (\log n)/2$ and $\varepsilon_t < \frac{1}{\log n}$.

Since $n_t\geq\sqrt{n}$, the function $\OR_{X_t}(x)$ evaluates to $1$ under the distribution $\mu_{1}^{X_t} = \mc{U}_{X_t}$ with probability $1-o(1)$. Thus, $q_t$ must also evaluate to $1$ on some input. However, since $q_t$ is a constant polynomial, this implies that $q_t = 1$. But this implies that $q_t(0,0,\ldots,0) = 1$ as well, which leads to a contradiction, since $q_t$ is obtained by setting some input bits of $q_0$ to $0$ and $q_0(0,0,\ldots,0) = 0$ by our choice of $q_0$. This completes the proof of the theorem.

Now we describe how to obtain the sequence $q_1,\ldots,q_t$. More precisely, we describe how to obtain $q_i$ from $q_{i-1}$ assuming $d_{i-1}\geq 1$. Fix any $i\geq 1$ such that $d_{i-1}\geq 1$. We assume that the sequence $q_1,\ldots,q_{i-1}$ of polynomials constructed so far satisfy the above properties. 

For brevity, let $q,X,m,d,\ell,\varepsilon$ denote $q_{i-1},X_{i-1},n_{i-1},d_{i-1}, \ell_{i-1}, \varepsilon_{i-1}$ respectively.

We know that $q$ is $(X,\ell,\varepsilon)$-good. As we did in \eqref{eq:qt-err} for $q_t$, we can use this to show that $\Err^{X}_1(q) < \frac{1}{2}$ and since $\OR_{X}(x)$ takes the value $1$ on an input $x\sim \mc{U}_{X}$ with probability $1-o(1)$, we see that
\begin{equation}
\label{eq:q-conc}
\prob{x\sim \mc{U}_{X}}{q(x) = 1} \geq \frac{1}{2}-o(1)\geq \frac{1}{3}.
\end{equation}

\cref{lem:NV} then implies that there cannot be $r$ disjoint monomials of degree $d$ in $q$. To see this, assume that there are indeed $r$ many disjoint monomials of degree $d$ in $q$. Then by \cref{lem:NV}, the probability that $q(x)-1 = 0$ for a random $x\sim \mc{U}_X$ is at most 
\begin{align*}
Bd^{4/3}r^{-\frac{1}{4d+1}}\sqrt{\log r} &\leq Bd_0^{4/3}r^{-\frac{1}{5d_0}}\sqrt{\log r}\\
&\leq Bd_0^{4/3}\cdot \frac{\sqrt{10d_0\log(d_0\log^2 n)}}{d_0^{2}\log^4 n} = o(1).
\end{align*}

This contradicts \eqref{eq:q-conc}.

Hence, we know that $q$ cannot be contain more than $r$ many disjoint monomials of degree $d$. Let $S$ be any maximal set of disjoint monomials appearing in $q$. Note that by definition, every monomial of degree $d$ contains at least one variable from $S$ and hence setting all the variables in $S$ reduces the degree of the polynomial. The number of variables appearing in $S$ is at most $d|S|\leq dr$. 

We now choose a random zero-fixing restriction $\rho$ with $*$-probability $p$ as defined above and consider the polynomial $q|_\rho$. Define the following ``bad'' events:

\begin{itemize}
\item $\mc{E}_1(\rho)$ is the event that $|X_\rho| \not\in [pm/2,3pm/2]$.
\item $\mc{E}_2(\rho)$ is the event that some variable in $S$ is not set to $0$.
\item $\mc{E}_3(\rho)$ is the event that $q|_\rho$ is not $(X_\rho,\ell',\varepsilon')$-good where $\ell' = \ell-b$ and $\varepsilon' = \varepsilon\cdot \exp\left(\frac{16b}{\log n}\right)$.
\end{itemize}

We claim that there is a $\rho$ so that none of the bad events $\mc{E}_1(\rho),\mc{E}_2(\rho)$ or $\mc{E}_3(\rho)$ occur. This will imply that we can take $q_{i} = q|_\rho, X_i = X_\rho, \ell_i = \ell', \varepsilon_i = \varepsilon'$ and we will be done. So we only need to show that $\prob{\rho}{\mc{E}_1(\rho) \vee \mc{E}_2(\rho) \vee \mc{E}_3(\rho)}<1$. This is done as follows.

\begin{itemize}
\item $\prob{\rho}{\mc{E}_1(\rho)}$: By \cref{clm:calc}, we know that $m \geq \sqrt{n}$ and hence $\avg{\rho}{|X_\rho|} = pm = m\cdot \frac{1}{n^{o(1)}} \geq n^{1/4}$. Hence, by a Chernoff bound, the probability that $|X_\rho|\not\in [pm/2,3pm/2]$ is bounded by $\exp\left(-\Omega(n^{1/4})\right)$.

\item $\prob{\rho}{\mc{E}_2(\rho)}$: By a union bound over $S$, this probability is bounded by $p|S|\leq rd_0/r^2 < \frac{1}{\log n}$.

\item $\prob{\rho}{\mc{E}_3(\rho)}$: By \cref{obs:restr-err}, we know that for any $i$,
\[
\avg{\rho}{\Err^{X_\rho}_i(q|_\rho)} = \Err^{X}_{i+b}(q).
\]
Hence, 
\begin{equation}
\label{eq:err-qrho-q}
\avg{\rho}{\avg{i\in [\ell']}{\Err^{X_\rho}_i(q|_\rho)}} = \avg{i\in [\ell']}{\Err^X_{i+b}(q)} = \avg{i\in \{b+1,\ldots, b+\ell'\}}{\Err^X_i(q)} = \avg{i\in \{b+1,\ldots, \ell\}}{\Err^X_i(q)}.
\end{equation}

We can bound the right hand side of the above equation by
\[
\avg{i\in \{b+1,\ldots, \ell\}}{\Err^X_i(q)} \leq \frac{1}{(1-\frac{b}{\ell})}\avg{i\in [\ell]}{\Err^X_i(q)} \leq \frac{\varepsilon}{(1-\frac{b}{\ell})}
\]
where the final inequality follows from the fact that $q$ is $(X,\ell,\varepsilon)$-good. Further, by \cref{clm:calc}, we know that $\ell \geq \frac{\log n}{4} \gg b$, and hence we can bound the above as follows.
\[
\avg{i\in \{b+1,\ldots, \ell\}}{\Err^X_i(q)} \leq \frac{\varepsilon}{(1-\frac{b}{\ell})} \leq \varepsilon\cdot (1+\frac{2b}{\ell})\leq \varepsilon\cdot (1+\frac{8b}{\log n}).
\]
Plugging the above bound into \eqref{eq:err-qrho-q}, we obtain
\[
\avg{\rho}{\avg{i\in [\ell']}{\Err^{X_\rho}_i(q|_\rho)}} \leq \varepsilon\cdot (1+\frac{8b}{\log n})\leq \varepsilon\cdot \exp\left(\frac{8b}{\log n}\right).
\]
By Markov's inequality, 
\[
\prob{\rho}{\avg{i\in [\ell']}{\Err^{X_\rho}_i(q|_\rho)} > \varepsilon \cdot \exp\left(\frac{16b}{\log n}\right)} \leq \exp\left(-\frac{8b}{\log n}\right() = 1 - \Omega(\frac{b}{\log n}) \leq 1-\frac{2}{\log n}.
\]
Thus, $\prob{\rho}{\mc{E}_3(\rho)}\leq 1-\frac{2}{\log n}$.
\end{itemize}

By a union bound, we have
\[
\prob{\rho}{\mc{E}_1(\rho) \vee \mc{E}_2(\rho) \vee \mc{E}_3(\rho)} \leq \exp\left(-\Omega(n^{1/4})\right) + \frac{1}{\log n} + 1- \frac{2}{\log n} < 1.
\]

\section{Open questions}
\label{sec:open}

Both our results leave some scope for improvement.

\paragraph{Independence required to fool $\AC^0$.} A close inspection of our proof (including the details of \cref{lem:Tarui-Brav} and \cref{thm:bravtal}) shows that $(\log s)^{3d+O(1)}\cdot \log(1/\varepsilon)$-wise independence is sufficient to $\varepsilon$-fool $\AC^0(s,d)$. Avishay Tal (personal communication) showed that this can be further improved to $(\log s)^{2.5d+O(1)}\cdot \log(1/\varepsilon)$-wise independence. It is open if this can be strengthened to, say, $(\log s)^{d+O(1)}\cdot \log(1/\varepsilon)$ or even $(\log s)^{d-1}\cdot \log(1/\varepsilon)$, matching the lower bound due to Mansour~\cite{LubyV1996}.

\paragraph{Probabilistic degree of $\OR$.} It remains an open question to prove tight bounds on the degree of any $\varepsilon$-error probabilistic polynomial for the $\OR$ function on $n$ variables for any $n$ and $\varepsilon.$ The ideal result in this direction would be a lower bound of $\Omega(\log n\cdot \log(1/\varepsilon))$, matching the upper bounds from~\cite{BeigelRS1991} and~\cite{Tarui1993} mentioned in the Introduction. A result of Alon, Bar-Noy, Linial, and Peleg~\cite{AlonBLP1991} implies that for polynomials of the specific form used in~\cite{BeigelRS1991,Tarui1993}\footnote{The polynomials from~\cite{BeigelRS1991,Tarui1993} are of the form $P(x_1,\ldots,x_n) = 1-\prod_{S\in \mc{F}}(1-\sum_{i\in S} x_i)$ for some family $\mc{F}$ of subsets of $[n].$ The lower bound of Alon et al.~\cite{AlonBLP1991} holds for polynomials of this form.} and $\varepsilon = 1/n,$ the degree bound of $O(\log^2 n)$ is tight.

\section{Acknowledgements} 

We thank Swagato Sanyal and Madhu Sudan for encouragement and useful discussions  which greatly simplified our proofs. We thank Avishay Tal for his generous feedback and comments and also for showing us the improvement in seedlength mentioned in \cref{sec:open}. We  thank Paul Beame and Xin Yang for pointing out that a change in parameters results in a quantitative improvement in the lower bound obtained in~\cref{thm:OR-lbd}. We thank Noga Alon for pointing out the implications of~\cite{AlonBLP1991} to our setting. Finally, we thank the anonymous reviewers of RANDOM 2016 and Random Structures \& Algorithms for their careful perusal of our paper.

{\small
\bibliographystyle{plainurl}
\bibliography{ac0-poly-bib}
}

\appendix

\section{The limitations of the Nisan-Wigderson paradigm}
\label{app:NW}

In this section, we show that the general hardness-to-randomness tradeoff of Nisan and Wigderson~\cite{NisanW1994} does not yield a PRG with optimal seedlength as a function of $\varepsilon$ \emph{given our current knowledge} of circuit lower bounds.

We start by describing the meta-result of Nisan and Wigderson~\cite{NisanW1994} that allows us to convert any sufficiently hard-to-compute function for a class of circuits to a PRG for a slightly weaker class of circuits. The result is true in greater generality than we describe here but to keep things concrete, we stick to the setting of $\AC^0(s,d)$.

We say that a function $f:\{0,1\}^r\rightarrow \{0,1\}$ is $(s,d,\varepsilon)$-hard if given any circuit $C$ from $\AC^0(s,d)$ of size $s$, we have
\[
\prob{x\in \{0,1\}^r}{C(x) = f(x)} \leq \frac{1}{2}+\varepsilon.
\]

For non-negative integers $m,r,\ell,s$, we say that a family $\mc{F}\subseteq \binom{[m]}{r}$, we say that $\mc{F}$ is an $(m,r,\ell,s)$ design if $|\mc{F}| = s$ and for any distinct $S,T\in \mc{F}$, we have $|S\cap T|\leq \ell$.

Nisan and Wigderson~\cite{NisanW1994} show the following.

\begin{theorem}[{\cite{NisanW1994}}]
\label{thm:NW}
Let $m,r,\ell,s\in \mathbb{N}$ be positive parameters such that $m\geq r\geq \ell$. Given an explicit $f:\{0,1\}^r\rightarrow\{0,1\}$ that is $(s\cdot 2^\ell, d+1,\varepsilon/s)$-hard and an explicit $(m,r,\ell,s)$-design, we can construct an explicit PRG $G:\{0,1\}^m\rightarrow \{0,1\}^s$ that fools circuits from $\AC^0(s,d)$ with error at most $\varepsilon$.
\end{theorem}

To use this theorem, we need a hard function for circuits in $\AC^0$. The best such result known currently is the following due to Impagliazzo, Matthews, and Paturi~\cite{ImpagliazzoMP2012} (see also H\r{a}stad~\cite{Hastad2014}).

\begin{theorem}
\label{thm:IMP}
Let $d\geq 1$ be a constant. The Parity function on $r$ is bits is $(s_1,d_1,\delta)$-hard if $r \geq A(\log s_1)^{d_1-1}\cdot \log(1/\delta)$ for some constant $A > 0$ depending on $d$.
\end{theorem}

Thus, if we want to apply \cref{thm:NW} alongside the lower bound given by \cref{thm:IMP} to construct PRGs that $\varepsilon$-fool $\AC^0(s,d)$, then we need 
\begin{equation}
\label{eq:ineq-r}
r\geq A(\log s + \ell)^{d}\cdot \log (s/\varepsilon)\geq A(\log s + \ell)^{d}\cdot \log (1/\varepsilon)
\end{equation}
for some constant $A > 0$ depending on $d$.

Further, to construct an $(m,r,\ell,s)$-design, we claim that we further need
\begin{equation}
\label{eq:ineq-m}
m\geq \min\{r^2/2\ell, s\}.
\end{equation}

We justify \eqref{eq:ineq-m} below, but first we use it to prove that the Nisan-Wigderson paradigm cannot be used to obtain seedlength optimal in terms of $\varepsilon$ for a large range of $\varepsilon$. 

We assume that $\varepsilon \geq \exp\left(-s^{1/4}\right)$ (the same proof works as long as $\varepsilon \geq \exp\left(-s^{\frac{1}{2}-\Omega(1)}\right)$). In this setting, we show that $m \geq B(\log s)^{2d-1}\cdot (\log (1/\varepsilon))^2$ for some constant $B$ depending on $d$.  

To see this, note that if $m\geq s$, then trivially we have $(\log s)^{2d-1}\cdot (\log 1/\varepsilon)^2 \leq s^{\frac{1}{2}+o(1)} < s\leq m$. So we assume that $m < s$.

In this case, \eqref{eq:ineq-m} tells us that $m\geq r^2/2\ell$, which yields

\begin{align*}
m &\geq \frac{r^2}{2\ell} \geq \frac{A^2 (\log s+\ell)^{2d}\cdot (\log 1/\varepsilon)^2}{2\ell}\\
&\geq \frac{A^2 (\log s)^{2d-1}\ell (\log 1/\varepsilon)^2}{2\ell} = \Omega(A^2 (\log s)^{2d-1}\cdot (\log(1/\varepsilon))^2)
\end{align*}

as required. 

The inequality \eqref{eq:ineq-m} is a standard combinatorial fact and can be found in many standard textbooks. For completeness, here is a simple proof using inclusion-exclusion. 

Note that if $s\leq r$, then we immediately have $m\geq r \geq s$ and \eqref{eq:ineq-m} is proved. So assume that $s > r$ and in particular given any $(m,r,\ell,s)$-design $\mc{F}$, we can choose $t = r/\ell$ sets $T_1,\ldots, T_t$ from $\mc{F}$. By inclusion-exclusion, we have
\begin{align*}
m &\geq |\bigcup_{i\in [t]}T_i| \geq \sum_i |T_i| - \sum_{i < j}|T_i\cap T_j|\\
 &\geq rt - \frac{t^2}{2}\cdot \ell \geq \frac{r^2}{2\ell}
\end{align*}
which concludes the proof of \eqref{eq:ineq-m}.

\end{document}